\documentclass[onecolumn,superscriptaddress,aps,pra]{revtex4-2}

\usepackage{graphicx}
\usepackage{amsmath}
\usepackage{amsthm}
\usepackage{amssymb}
\usepackage{latexsym}
\usepackage{array}
\usepackage{hyperref}
\usepackage{amsfonts}
\usepackage{dsfont}
\usepackage{mathrsfs}
\usepackage{verbatim}
\usepackage{bbold}
\usepackage[normalem]{ulem}
\usepackage{upgreek}
\usepackage{makecell}
\usepackage{adjustbox}
\usepackage{algorithm}
\usepackage{algpseudocode}
\usepackage{color}
\usepackage{bm}
\usepackage{times}
\usepackage{booktabs}
\usepackage{multirow}

\newcommand{\bra}[1]{\left<#1\right|}
\newcommand{\ket}[1]{\left|#1\right>}
\newcommand{\abs}[1]{\left|#1\right|}

\newcommand{\braket}[2]{\left<{#1}|{#2}\right>}

\newcommand{\Var}{\mathrm{Var}}
\newcommand{\Prob}{\mathrm{Pr}}

\newtheorem{theorem}{Theorem}
\newtheorem{proposition}{Proposition}

\newtheorem{corollary}{Corollary}

\begin{document}

\title{Single-shot measurement learning as a self-certifying estimator for quantum-enhanced sensing}

\author{Jeongho~Bang}\email{jbang@yonsei.ac.kr}
\affiliation{Institute for Convergence Research and Education in Advanced Technology, Yonsei University, Seoul 03722, Republic of Korea}
\affiliation{Department of Quantum Information, Yonsei University, Incheon 21983, Republic of Korea}

\date{\today}

\begin{abstract}
Single-shot measurement learning (SSML) learns a compensation unitary from a one-bit success/failure record and halts after a prescribed run of consecutive successes. We recast SSML as an adaptive estimator on a parameterized sensing manifold and ask what role it can play in quantum-enhanced sensing. First, we show that the terminal run itself furnishes an intrinsic certificate of local alignment: longer terminal runs certify smaller infidelity, and near the optimum this becomes a Fisher-calibrated certificate of parameter error. Second, for compensation-type sensing families, the Bernoulli success/failure record is locally matched to the probe quantum Fisher information (QFI), so SSML preserves the probe's metrological content despite using only one classical bit per copy. In this sense, SSML makes the quantum enhancement carried by the probe operationally available in an online self-terminating protocol. Applied to GHZ/NOON probes of depth $m$, SSML retains the familiar square-root entanglement gain over product probes at fixed total resource, while an ideal multiscale architecture remains compatible with Heisenberg scaling. Monte Carlo simulations of photonic NOON-state phase sensing show the expected near-inverse decay of terminal infidelity with entangled shots, SQL-like total-resource scaling at fixed entanglement depth, the corresponding fixed-resource entanglement gain, the global limitation of a single fringe scale, and the recovery of Heisenberg-compatible behavior under ideal multiscale hand-off. These results identify SSML as a Fisher-preserving, self-certifying estimator layer for quantum-enhanced sensing.
\end{abstract}

\maketitle

%---------------------------------------------------------------------------------------------------------------------------------
\section*{Introduction}
%---------------------------------------------------------------------------------------------------------------------------------

Quantum sensing is usually introduced as a story about probe states. Product probes saturate the standard quantum limit (SQL), entangled probes can increase the quantum Fisher information (QFI), and ideal multiscale or error-free strategies can approach the Heisenberg resource law~\cite{Giovannetti2004,Giovannetti2006,Degen2017,Braun2018,Pezze2018}. Yet, the probe design is only half of the metrological narrative. In a laboratory, the experimenter must also decide how to update controls from a stream of measurement outcomes, how to tell whether a lock has been reached, and when the protocol has acquired enough evidence to stop. A metrological architecture is therefore not specified by its QFI alone; it also needs an estimator layer that can turn the measurement outcomes into a reliable and operationally certified answer~\cite{Paris2009,Wiseman1995,Higgins2007,Hentschel2010}.

Single-shot measurement learning (SSML) is a particularly economical candidate for such an estimator layer~\cite{Lee2018,Lee2021,Bang2026run}. Each fresh copy of an unknown state is processed by a tunable unitary, projected onto a fixed fiducial test, and reduced to a single classical bit: success or failure. The controller stores only a tiny internal memory, the counter $M_S$ of consecutive successes. It updates the control only after failures and halts when $M_S=M_H$. In its original state-learning formulation, SSML was numerically shown to achieve near-$O(N^{-1})$ average infidelity~\cite{Lee2018}; a later linear-optical experiment demonstrated almost ideal scaling $O(N^{-0.983})$ together with the empirical usefulness of the monitored proxy $(1+M_S)^{-1}$~\cite{Lee2021}. These results are striking because the protocol is simultaneously one-bit, online, and hardware-light~\cite{Bang2026run}.

Despite this, the metrological meaning of SSML has remained somewhat implicit. The original SSML studies were focused on the state-learning perspective. Quantum sensing, however, asks a related but not identical question. In sensing, one is typically interested not in reconstructing an arbitrary state in full Hilbert space, but in estimating a parameter that labels a physically structured state family. This shift in viewpoint raises sharper questions. What exactly is certified when SSML halts? Why should a terminal run of consecutive successes count as statistical evidence rather than merely as a heuristic stopping flag? Does the one-bit record preserve the Fisher information already carried by the probe, or does aggressive compression into a yes/no stream wash out the very metrological resource one hoped to exploit? Finally, if entangled probes are inserted into the protocol, is their quantum-metrological gain retained or lost inside the learning loop?

The aim of this work is to answer these questions. We reinterpret SSML as an adaptive sensing protocol on a parameterized manifold of probe states. In the quantum sensing language, the protocol has a natural two-stage meaning: an acquisition stage, which finds the correct branch or basin of the sensing landscape, and a lock stage, which certifies local alignment~\cite{Bang2026run}. The central observation is that the halting rule is itself an intrinsic sequential certificate. Once a terminal run of length $M_H$ has been observed, the run probability under a low-fidelity null hypothesis is exponentially suppressed, so the terminal run immediately defines an operational certificate scale. Near the optimum, the usual local QFI/Bures geometry converts this infidelity scale into a parameter error bar.

A second theme of this work concerns the role of classical information~\cite{Braunstein1994,Paris2009}. A one-bit outcome per copy may sound too coarse to support precise metrology, but this intuition is misleading in the local regime relevant for sensing. For compensation-type sensing families~\cite{Wiseman1995,Berry2000,Higgins2007,Xiang2011}, the Bernoulli success/failure record generated by SSML is locally QFI-matched: its classical Fisher information approaches the probe QFI itself. In other words, SSML does not throw away the sensing advantage present in the probe. This point is conceptually important, because it separates the resource-bearing part of the protocol, namely the probe family, from the estimator layer, namely the adaptive rule and the stopping-time semantics.

To make the above-described viewpoints concrete, we study a photonic NOON-state phase-sensing example~\cite{Bollinger1996,Dowling2008}. We show numerically that in a local branch-resolved setting, the terminal infidelity follows the familiar $\nu^{-1}$ law as a function of entangled shots $\nu=\mathbb{E}[T]$, while the phase root-mean-square error (RMSE) displays the expected SQL slope in the total resource $R=m\,\mathbb{E}[T]$ together with the preserved $\sqrt{m}$ quantum gain at fixed $R$. We then explain why a single entangled scale is globally ambiguous and why a multiscale design is the natural route to Heisenberg-compatible behavior.

%---------------------------------------------------------------------------------------------------------------------------------
\section*{Results}
%---------------------------------------------------------------------------------------------------------------------------------

%-------------------------------------------------------------------------------------------------------------------------------------------------------------------------------------------
\subsection*{A brief introduction to the SSML algorithm}

We briefly recall SSML in its original state-learning form~\cite{Lee2018,Lee2021,Bang2026run}. In SSML, repeated copies of a pure state $\ket{\psi_\tau}$ are processed one by one. On each shot, a tunable unitary $\hat{U}(\bm{p})$ is applied and the output is tested by a fixed yes/no projection onto a fiducial state $\ket{\phi_0}$. The corresponding success probability is
\begin{eqnarray}
p_s(\bm{p}) = \abs{\bra{\phi_0}\hat{U}(\bm{p})\ket{\psi_\tau}}^2,
\label{eq:ssml-generic-success}
\end{eqnarray}
and the associated infidelity is
\begin{eqnarray}
\epsilon(\bm{p}) := 1-p_s(\bm{p}).
\label{eq:ssml-generic-infidelity}
\end{eqnarray}
The learning task is therefore to tune the control $\bm{p}$, so that the input state is mapped as close as possible to the fiducial target.

The protocol stores only a very small internal memory: the current control $\bm{p}_n$ and the counter $M_S^{(n)}$ of consecutive successes. After a success, the controller keeps the current setting; otherwise, after a failure, it resets the counter and applies a randomized update, i.e.,
\begin{eqnarray}
M_S^{(n)}=
\begin{cases}
M_S^{(n-1)}+1, & \text{success},\\
0, & \text{failure},
\end{cases}
\quad
\bm{p}_{n+1}=
\begin{cases}
\bm{p}_n, & \text{success},\\
\bm{p}_n+\omega_n \bm{r}_n, & \text{failure},
\end{cases}
\label{eq:ssml-rule}
\end{eqnarray}
with step size $\omega_n = a(M_S^{(n-1)}+1)^{-b}$ and random direction $\bm{r}_n$. The protocol halts at
\begin{eqnarray}
T := \inf\{n\ge 1 : M_S^{(n)} = M_H\}.
\label{eq:halting-time}
\end{eqnarray}
This compact rule already contains the two structural ingredients that matter later for sensing: an acquisition part, driven by failure-triggered exploration, and a certification part, driven by the terminal run of consecutive successes. The sensing reinterpretation below keeps exactly this algorithmic backbone while changing the physical meaning of the control variable and of the success event.

%-------------------------------------------------------------------------------------------------------------------------------------------------------------------------------------------
\subsection*{SSML on a sensing manifold: from state learning to parameter locking}

We now specialize the generic SSML template to a one-parameter sensing family
\begin{eqnarray}
\ket{\psi^{(m)}_\lambda} = \hat{V}_\lambda \ket{\psi^{(m)}_0},
\label{eq:sensing-family}
\end{eqnarray}
where $\lambda$ is the parameter of interest and $m$ is the number of particles used in a single shot. The label $m=1$ corresponds to a product probe; $m>1$ may represent an entangled block, for example, a GHZ/NOON probe~\cite{Bollinger1996,Dowling2008}. Rather than treating the output as an arbitrary unknown state in the full Hilbert space, we restrict attention to the structured manifold generated by the known family $\hat{V}_\lambda$. The generic control variable $\bm{p}$ of Eq.~(\ref{eq:ssml-rule}) is correspondingly reinterpreted as a trial compensation parameter $\tilde{\lambda}$ through a unitary $\hat{U}(\tilde{\lambda})$ that attempts to undo the unknown sensing transformation and return the probe toward a fixed fiducial state.

Specializing Eq.~(\ref{eq:ssml-generic-success}) to this family gives the single-shot return probability
\begin{eqnarray}
 p_s(\tilde{\lambda},\lambda) = \abs{\bra{\psi^{(m)}_0}\hat{U}(\tilde{\lambda})\hat{V}_\lambda\ket{\psi^{(m)}_0}}^2,
 \quad
 \epsilon(\tilde{\lambda},\lambda) := 1-p_s(\tilde{\lambda},\lambda).
 \label{eq:ps-def}
\end{eqnarray}
When $\tilde{\lambda}$ approaches $\lambda$ inside the correct local basin, the compensated output $\hat{U}(\tilde{\lambda})\hat{V}_\lambda\ket{\psi^{(m)}_0}$ returns toward the fiducial probe and the success probability tends to one. In this way, the abstract state-learning objective of SSML becomes a physically meaningful lock condition: maximizing the success probability is equivalent to maximizing the probability of having correctly compensated the sensed signal.

The update law itself is unchanged in form. In the present one-parameter setting, the generic random direction $\bm{r}_n$ of Eq.~(\ref{eq:ssml-rule}) reduces to a scalar random increment $r_n$, so a success leaves $\tilde{\lambda}$ untouched whereas a failure triggers a randomized search step of size $\omega_n$. What changes is the interpretation of the trajectory. Early failures are not merely algorithmic setbacks; they correspond to exploration across branches or basins of the sensing landscape. Once the search enters the correct branch, long success runs indicate that the compensation unitary is staying near the locally optimal inverse. The original search--certification split of SSML (originally analyzed in a state-learning viewpoint~\cite{Bang2026run}) therefore becomes, in sensing language, an acquisition--lock split. The schematic of such an adaptive SSML sensing is presented in Fig.~\ref{fig:schematic}.

This reinterpretation is more than a change of the framework language. At the algorithmic level, SSML remains a one-bit adaptive loop with a randomized failure update and a run-length stopping rule. At the physical level, however, the loop is now tied to a structured parameter manifold whose local geometry can be quantified by the QFI~\cite{Paris2009}. That bridge is what allows the halting counter to become, in the next subsection, a metrological certificate rather than a purely heuristic stopping flag.

\begin{figure}[t]
    \centering
    \includegraphics[width=0.90\textwidth]{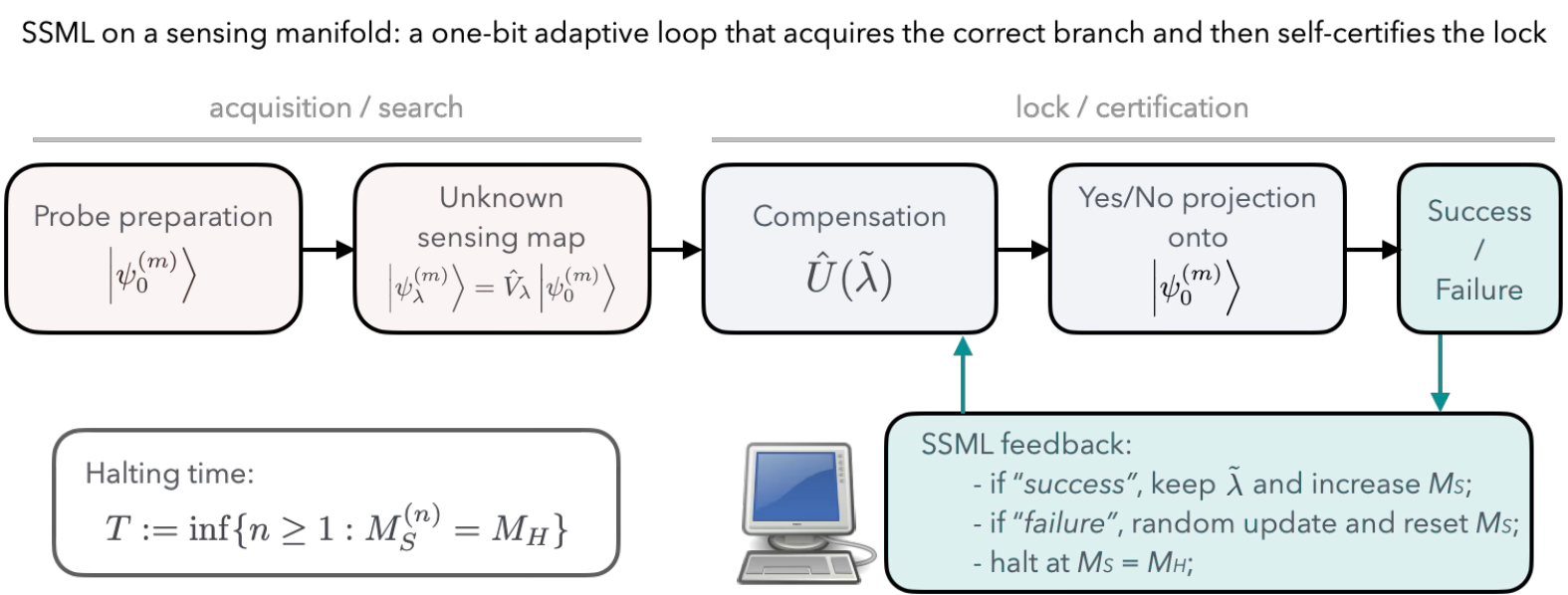}
    \caption{Schematic: SSML as an adaptive sensing loop. The unknown sensing map $\hat{V}_\lambda$ acts on a probe family, a compensation $\hat{U}(\tilde{\lambda})$ is learned shot by shot, and a yes/no projection onto the fiducial probe feeds a one-bit record back to the controller. The same protocol naturally separates into acquisition/search and lock/certification stages.}
    \label{fig:schematic}
\end{figure}

%-------------------------------------------------------------------------------------------------------------------------------------------------------------------------------------------
\subsection*{Terminal runs as sequential evidence: from halting rules to certificates}

The first structural question is deceptively simple: what does it mean, statistically, that SSML stops? The answer is encoded in a basic but important observation~\cite{Bang2026run}. During the terminal run of consecutive successes, the control is frozen. Therefore, the last $M_H$ outcomes are not produced under a changing adaptive policy; they are generated by repeated Bernoulli trials with one and the same success probability, namely the success probability of the final control. This makes the terminal run a valid object of the statistical inference.

Suppose that the final control corresponds to infidelity $\epsilon=1-p_s$. Then, the probability of observing $M_H$ consecutive successes is exactly
\begin{eqnarray}
\Prob(\text{terminal run of length }M_H\mid \epsilon)=(1-\epsilon)^{M_H}.
\label{eq:run-probability}
\end{eqnarray}
Now consider the null hypothesis that the final infidelity is not yet small, $\epsilon\ge \epsilon_0$. Since $(1-\epsilon)^{M_H}$ decreases monotonically with $\epsilon$, the most favorable case for the null is $\epsilon=\epsilon_0$. Hence,
\begin{eqnarray}
\Prob(\text{terminal run of length }M_H\mid \epsilon\ge \epsilon_0) \le (1-\epsilon_0)^{M_H}.
\label{eq:null-test}
\end{eqnarray}
If one wishes to reject this null at significance level $\eta$, it suffices to demand
\begin{eqnarray}
(1-\epsilon_0)^{M_H} \le \eta.
\label{eq:null-threshold}
\end{eqnarray}
Solving for $\epsilon_0$ produces the intrinsic certificate scale
\begin{eqnarray}
\epsilon_{\rm cert}(M_H,\eta) := 1-\eta^{1/M_H}.
\label{eq:eps-cert-exact}
\end{eqnarray}
For large $M_H$, a simple expansion gives
\begin{eqnarray}
\epsilon_{\rm cert}(M_H,\eta) = \frac{\ln(1/\eta)}{M_H} + O(M_H^{-2}).
\label{eq:eps-cert-asym}
\end{eqnarray}
Thus, $M_H$ is not merely a stopping threshold. It is an evidence budget stored in the run length itself.

This point is worth emphasizing because it explains the origin of the $1/N$-type law in a way that is independent of any heavy optimal-estimation formalism. A protocol that halts only after witnessing $M_H$ consecutive successes cannot certify an infidelity scale that decreases parametrically faster than $1/M_H$. Since every halted trajectory consumes at least $M_H$ copies, the terminal-run semantics already singles out the exponent $1$ at the level of certifiable infidelity. The later Fisher analysis does not replace this observation; it explains why the same exponent reappears when the same problem is viewed through local information geometry.

To translate the certificate into a metrological error bar, we next use the local geometry of a smooth pure-state family. For a one-parameter manifold, the overlap between neighboring states obeys the standard QFI/Bures expansion \cite{Braunstein1994,Paris2009}
\begin{eqnarray}
\abs{\braket{\psi^{(m)}_\lambda}{\psi^{(m)}_{\lambda+\delta\lambda}}}^2 = 1 - \frac{F_Q(\lambda)}{4}(\delta\lambda)^2 + O((\delta\lambda)^4).
\label{eq:bures-qfi}
\end{eqnarray}
In a compensation family, the same local law governs the return probability, so near the optimum and inside the correct branch,
\begin{eqnarray}
\epsilon(\tilde{\lambda},\lambda) = 1 - p_s(\tilde{\lambda},\lambda) = \frac{F_Q(\lambda)}{4}(\tilde{\lambda}-\lambda)^2 + O\!\left((\tilde{\lambda}-\lambda)^4\right).
\label{eq:qfi-local-geometry}
\end{eqnarray}
Eq.~(\ref{eq:eps-cert-exact}) and Eq.~(\ref{eq:qfi-local-geometry}) immediately imply the following.

\begin{proposition}[Fisher-calibrated parameter certificate]
\label{prop:param-certificate}
Near the optimum of a pure-state sensing manifold, the SSML halting rule $M_S=M_H$ implies the parameter certificate
\begin{eqnarray}
\abs{\tilde{\lambda}-\lambda}_{\rm cert} \le \frac{2}{\sqrt{F_Q(\lambda)}}\sqrt{1-\eta^{1/M_H}} \approx \frac{2}{\sqrt{F_Q(\lambda)}}\sqrt{\frac{\ln(1/\eta)}{M_H}}.
\label{eq:param-certificate}
\end{eqnarray}
\end{proposition}

\begin{proof}[Proof sketch]---Eq.~(\ref{eq:null-threshold}) tells us that a terminal run of length $M_H$ rules out any infidelity larger than $\epsilon_{\rm cert}(M_H,\eta)$ at significance level $\eta$. In the local regime, Eq.~(\ref{eq:qfi-local-geometry}) identifies infidelity with the squared QFI distance up to quartic corrections. Taking square roots yields Eq.~(\ref{eq:param-certificate}).
\end{proof}

\begin{figure}[t]
    \centering
    \includegraphics[width=1.00\textwidth]{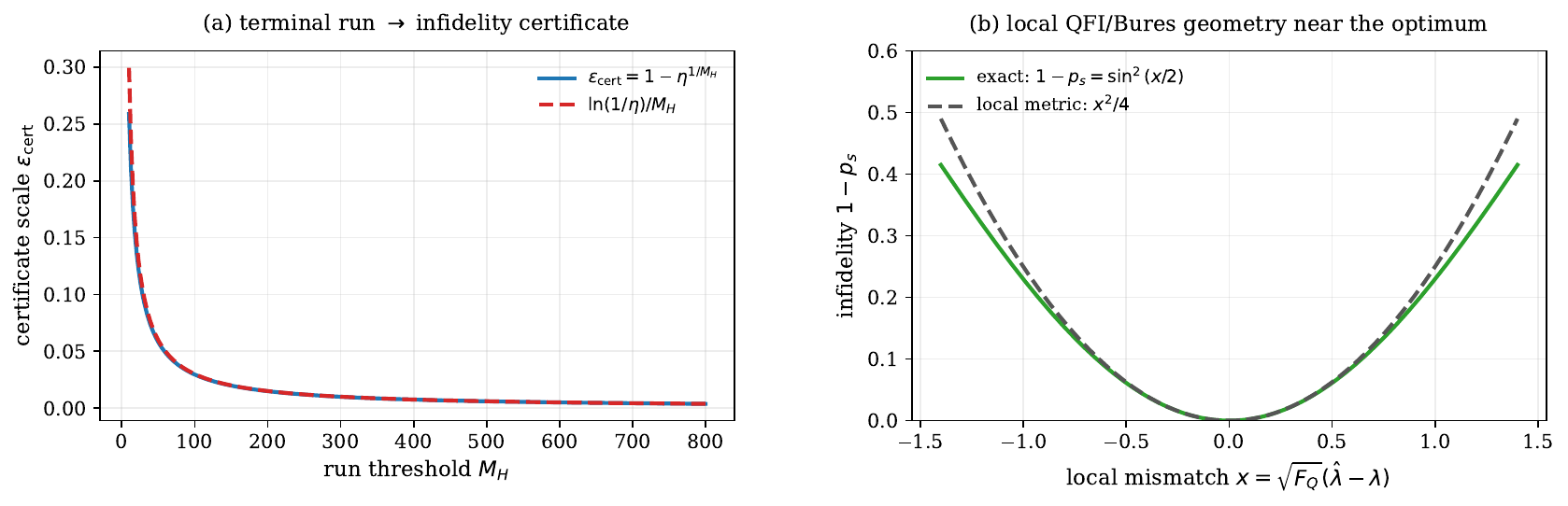}
    \caption{Certificate and local geometry. (a) The terminal run $M_H$ induces the intrinsic certificate scale $\epsilon_{\rm cert}=1-\eta^{1/M_H}$, with the expected $\ln(1/\eta)/M_H$ asymptote. (b) Near the optimum, the infidelity of a compensation family is locally quadratic, $1-p_s \approx x^2/4$ with $x=\sqrt{F_Q}(\tilde{\lambda}-\lambda)$, i.e., the standard QFI/Bures metric relation.}
    \label{fig:cert-fisher}
\end{figure}

Proposition~\ref{prop:param-certificate} is the first main bridge between the learning and sensing languages. It says that the halting counter $M_H$, a purely algorithmic object in the original SSML formulation, is also a metrological confidence dial. Each additional consecutive success tightens the certified parameter interval, and it does so online, without requiring a separate tomography step after the protocol stops. Fig.~\ref{fig:cert-fisher}(a) visualizes the exact certificate scale and its $1/M_H$ asymptote, while Fig.~\ref{fig:cert-fisher}(b) illustrates the local quadratic geometry behind the conversion from infidelity to parameter error.

%-------------------------------------------------------------------------------------------------------------------------------------------------------------------------------------------
\subsection*{Why one bit is enough locally: SSML as a Fisher-preserving layer}

The second main question is whether the extreme compression of SSML, i.e., one bit per copy plus the short memory stored in $M_S$, destroys metrological content. The local answer is no. In fact, for the compensation families of interest, the Bernoulli success/failure record carries the same local Fisher information as the probe itself.

We introduce the local mismatch coordinate $x:=\tilde{\lambda}-\lambda$ and expand the success probability as
\begin{eqnarray}
 p_s(x) = 1-\frac{F_Q}{4}x^2 + O(x^4).
 \label{eq:local-ps}
\end{eqnarray}
Because each shot produces a Bernoulli outcome, the classical Fisher information of the one-bit record is~\cite{Paris2009}
\begin{eqnarray}
I_{\rm cl}(x) = \frac{[\partial_x p_s(x)]^2}{p_s(x)[1-p_s(x)]}.
\label{eq:Icl-def}
\end{eqnarray}
Substituting Eq.~(\ref{eq:local-ps}) gives a remarkably simple result:
\begin{eqnarray}
I_{\rm cl}(x) = F_Q + O(x^2).
\label{eq:Icl-local}
\end{eqnarray}
Thus, the one-bit record is locally QFI-matched.

\begin{theorem}[Local Fisher matching of SSML]
\label{thm:fisher-matching}
For a smooth compensation-type sensing family, the Bernoulli success/failure record generated by SSML preserves the probe's local Fisher information:
\begin{eqnarray}
\lim_{x\to 0} I_{\rm cl}(x) = F_Q.
\end{eqnarray}
Hence, the one-bit record is locally QFI-matched.
\end{theorem}

\begin{proof}[Proof sketch]---From Eq.~(\ref{eq:local-ps}), we have $\partial_x p_s(x) = -\frac{F_Q}{2}x+O(x^3)$. The denominator of Eq.~(\ref{eq:Icl-def}) is $p_s(x)[1-p_s(x)]=\frac{F_Q}{4}x^2+O(x^4)$, while the numerator is $[\partial_x p_s(x)]^2=\frac{F_Q^2}{4}x^2+O(x^4)$. Taking the ratio yields Eq.~(\ref{eq:Icl-local}).
\end{proof}

The theorem pinpoints the role of SSML in quantum-enhanced sensing. In summary, SSML is Fisher-preserving (not Fisher-generating). The QFI is set upstream by the probe family. What SSML contributes is an adaptive and experimentally lightweight controller that turns a sequence of one-bit outcomes into an online lock without discarding the local statistical content of the probe. This separation of roles matters conceptually: it tells us that the quantum advantage of an entanglement-based sensor and the operational convenience of SSML are compatible rather than competing design principles.

The same perspective also clarifies why the counter $M_S$ already behaves like a monitored precision indicator in earlier SSML work \cite{Lee2018,Lee2021}. If the controller is approximately stationary over a short local stretch of shots, then the run length $L$ before the next failure is geometric,
\begin{eqnarray}
\Prob(L=\ell\mid \epsilon)=(1-\epsilon)^\ell\epsilon,
\quad
\mathbb{E}[L\mid\epsilon]=\epsilon^{-1}-1.
\label{eq:geom-run}
\end{eqnarray}
Replacing $L$ by the observed counter $M_S$, we obtain the familiar proxy
\begin{eqnarray}
\epsilon \approx (1+M_S)^{-1}.
\label{eq:monitored-infidelity}
\end{eqnarray}
On a sensing manifold, Eq.~(\ref{eq:qfi-local-geometry}) then turns this into
\begin{eqnarray}
\Delta\lambda_{\rm mon} \sim \frac{2}{\sqrt{F_Q}}\frac{1}{\sqrt{1+M_S}},
\label{eq:monitoring-proxy}
\end{eqnarray}
so the same counter that drives the halting logic also serves as a ``real-time metrological monitor.''

%-------------------------------------------------------------------------------------------------------------------------------------------------------------------------------------------
\subsection*{Entangled probes, preserved quantum gain, and what Heisenberg compatibility really means}

We now turn to the question that matters most for quantum-enhanced sensing: what happens when the probes carry a genuine quantum advantage? Here, let us consider the canonical GHZ/NOON phase family~\cite{Bollinger1996,Dowling2008}
\begin{eqnarray}
\ket{\Psi^{(m)}_\theta} = \frac{\ket{0}^{\otimes m}+e^{i m\theta}\ket{1}^{\otimes m}}{\sqrt{2}}.
\label{eq:noon-state}
\end{eqnarray}
This family is generated by a phase shift with effective generator variance $m^2/4$, so its QFI is $F_Q=m^2$. Equivalently, with the compensation phase $\tilde{\theta}$ and fiducial projection onto $\ket{\Psi_0^{(m)}}$, the success probability takes the explicit form
\begin{eqnarray}
 p_s(\theta,\tilde{\theta})  = \cos^2\left[\frac{m(\theta-\tilde{\theta})}{2}\right].
 \label{eq:noon-ps}
\end{eqnarray}
Near the optimum,
\begin{eqnarray}
1-p_s \approx \frac{m^2}{4}(\theta-\tilde{\theta})^2,
\quad\text{and}\quad
F_Q = m^2.
\label{eq:noon-local}
\end{eqnarray}
This is the entangled analogue of Eq.~(\ref{eq:qfi-local-geometry}): the local mismatch is magnified by a factor $m$.

If $\nu$ entangled shots are consumed and the total particle budget is $R=m\nu$, then the Fisher-preserving character of SSML implies
\begin{eqnarray}
\mathrm{MSE}(\tilde{\theta}) &=& O\left(\frac{1}{\nu m^2}\right),
\quad
\Delta\theta = O\left(\frac{1}{m\sqrt{\nu}}\right) = O\left(\frac{1}{\sqrt{Rm}}\right).
\label{eq:fixed-m-gain}
\end{eqnarray}
The physical meaning of Eq.~\eqref{eq:fixed-m-gain} is subtle but important. For $m=1$, we recover the product-probe SQL, i.e., $\Delta\theta=O(R^{-1/2})$. For fixed entanglement depth $m>1$, the total-resource slope remains SQL-like, but the prefactor is improved by $\sqrt{m}$. This is a genuine quantum gain relative to the product probes at the same total resource, even though it is not yet Heisenberg scaling in the total resource itself.

\begin{corollary}[Quantum gain at fixed entanglement depth]
\label{cor:fixed-depth}
At fixed total resource $R$, SSML used with $m$-particle entangled probes preserves a $\sqrt{m}$ improvement over the product-state SQL:
\begin{eqnarray}
\Delta\theta_{\rm SSML-ent} = O\left(\frac{1}{\sqrt{Rm}}\right),
\quad
\Delta\theta_{\rm SQL} = O\left(R^{-1/2}\right).
\end{eqnarray}
\end{corollary}

\begin{proof}[Proof sketch]---The local Cram\'er--Rao scaling of an efficient estimator is set by the inverse total Fisher information. For $\nu$ independent entangled shots, the total Fisher information is $\nu F_Q=\nu m^2$, giving variance $O((\nu m^2)^{-1})$. Expressing $\nu$ as $R/m$, we obtain Eq.~(\ref{eq:fixed-m-gain}).
\end{proof}

The corollary explains the precise sense in which SSML can participate in quantum-enhanced sensing. SSML does not erase the $\sqrt{m}$ advantage. The estimator layer is compatible with the probe's quantum enhancement. At the same time, Eq.~(\ref{eq:noon-ps}) makes clear why entanglement-based sensing is not just a trivial plug-in replacement of the product case. The fringe period shrinks as $1/m$, so higher $m$ provides more local sensitivity but also more global ambiguity. This is exactly the setting in which the acquisition--lock interpretation of SSML becomes useful. A low-resolution stage can identify the correct fringe or branch, after which a high-resolution entangled stage can lock and certify within that branch~\cite{Higgins2007,Xiang2011}.

This naturally motivates a multiscale architecture. Let stage $j$ use entanglement depth $m_j=2^j$, and suppose that the output of stage $j$ is already branch-resolved before stage $j+1$ begins. If each stage drives the residual phase down to order $1/m_j$, then the number of entangled shots required at each stage is only $O(1)$ in the ideal local regime, whereas the particle cost of each shot is $m_j$. The stage cost is therefore $R_j=O(m_j)$ and the total cost is geometric,
\begin{eqnarray}
R_{\rm tot} = \sum_{j=0}^{J} O(m_j) = O(m_J),
\label{eq:geom-resource}
\end{eqnarray}
while the final residual uncertainty is $\Delta\theta_J=O(1/m_J)$. Hence,
\begin{eqnarray}
\Delta\theta_J = O(R_{\rm tot}^{-1}).
\label{eq:heisenberg-compatible}
\end{eqnarray}

\begin{proposition}[Ideal multiscale Heisenberg compatibility]
\label{prop:multiscale-hl}
An ideal branch-resolved SSML architecture with stage sizes $m_j=2^j$ is compatible with the Heisenberg resource law $\Delta\theta_J = O(R_{\rm tot}^{-1})$ in the noiseless local regime.
\end{proposition}

Proposition~\ref{prop:multiscale-hl} is intentionally modest in its claim. It is not a general theorem about arbitrary noisy Heisenberg scaling, nor does it say that SSML can restore the Heisenberg performance when the probe family itself lacks the corresponding QFI. Rather, it says that once a multiscale entangled sensor has been arranged so that each stage operates in the correct local basin, the SSML estimator layer is fully compatible with the asymptotic Heisenberg resource law. In this sense SSML is a natural adaptive controller for coarse-to-fine entangled sensing.

%-------------------------------------------------------------------------------------------------------------------------------------------------------------------------------------------
\subsection*{Photonic NOON-state phase sensing: numerical analysis}

\subsubsection*{Photonic NOON-state SSML: local fixed-depth simulation} %------------------------------------------------------------------------------------------------

For the numerical analysis, we work in the local and branch-resolved regime. Here, we isolate the role of SSML as an estimator, not to conflate that role with the independent problem of global fringe disambiguation. In this regime, it is convenient to use the metric coordinate
\begin{eqnarray}
x := m(\theta-\tilde{\theta}),
\label{eq:metric-coordinate}
\end{eqnarray}
for which the physical update $\tilde{\theta}\mapsto \tilde{\theta} +(a/m)(M_S+1)^{-b}r$ becomes the $m$-independent metric update $x\mapsto x+a(M_S+1)^{-b}r$. We use the successful parameter choices $a=0.3$ and $b=0.5$ inherited from the original SSML studies~\cite{Lee2018,Lee2021,Bang2026run}. For each pair $(m,M_H)$ with $m\in\{1,2,4,8\}$ and $M_H\in\{20,40,80,160,320,640\}$, we perform $10^4$ Monte Carlo trajectories.

The numerical results are summarized in Fig.~\ref{fig:noon-main}. The first panel (a) shows the mean terminal infidelity $\mathbb{E}[\epsilon_T]$ against the mean number of entangled shots $\nu=\mathbb{E}[T]$. The log-log slope is close to $-1$, with fit $\mathbb{E}[\epsilon_T]\propto \nu^{-0.95}$. This is the entangled-shot analogue of the familiar SSML infidelity law: terminal infidelity is governed by the stopping-time certificate structure, so the essential $1/\nu$ behavior persists regardless of the entanglement depth. The second panel (b) changes the horizontal axis from entangled shots to total photons $R=m\mathbb{E}[T]$ and plots the phase RMSE. Here, the fitted slope is $\Delta\theta\propto R^{-0.48}$, i.e., SQL-like in the total resource. The entangled curves are shifted downward relative to the product curve, which is exactly what Eq.~(\ref{eq:fixed-m-gain}) predicts for fixed entanglement depth. Finally, the third panel (c) plots the fixed-resource quantity $\sqrt{R}\Delta\theta$ against $m$ and finds the clean law $\sqrt{R}\Delta\theta\propto m^{-1/2}$. This is the numerical signature of the preserved $\sqrt{m}$ quantum advantage.

Two further observations are worth highlighting. First, the three panels are mutually consistent in a nontrivial way. The panel (a) probes the stopping-time law in terms of entangled shots, whereas panels (b) and (c) re-express the same data in the physically more meaningful total resource. The fact that one sees both the near-$1/\nu$ infidelity law and the $m^{-1/2}$ fixed-resource gain in the same dataset strongly supports the interpretation that the estimator layer preserves the metrological improvement. Second, the fitted slopes do not need to be exactly $-1$ and $-1/2$ to validate the theory. Finite run thresholds, local clipping of the update, and halting-time bias naturally perturb the exponents at the few-percent level. What matters is the clear emergence of the expected asymptotic families.

Fig.~\ref{fig:noon-main} establishes the local fixed-depth message of the paper: once the correct branch has been acquired, SSML preserves both the near-$1/\nu$ infidelity law and the fixed-resource entanglement gain expected from the probe family. This local result, however, should not be overread as a claim of globally unambiguous phase recovery from a single entangled scale. Because the NOON fringe pattern is periodic [as in Eq.~(\ref{eq:noon-ps})], low terminal infidelity can coexist with large global phase error if the estimator locks onto the wrong fringe. That complementary point, together with the multiscale resolution of the ambiguity, is shown next subsection.

\begin{figure}[t]
    \centering
    \includegraphics[width=1.00\textwidth]{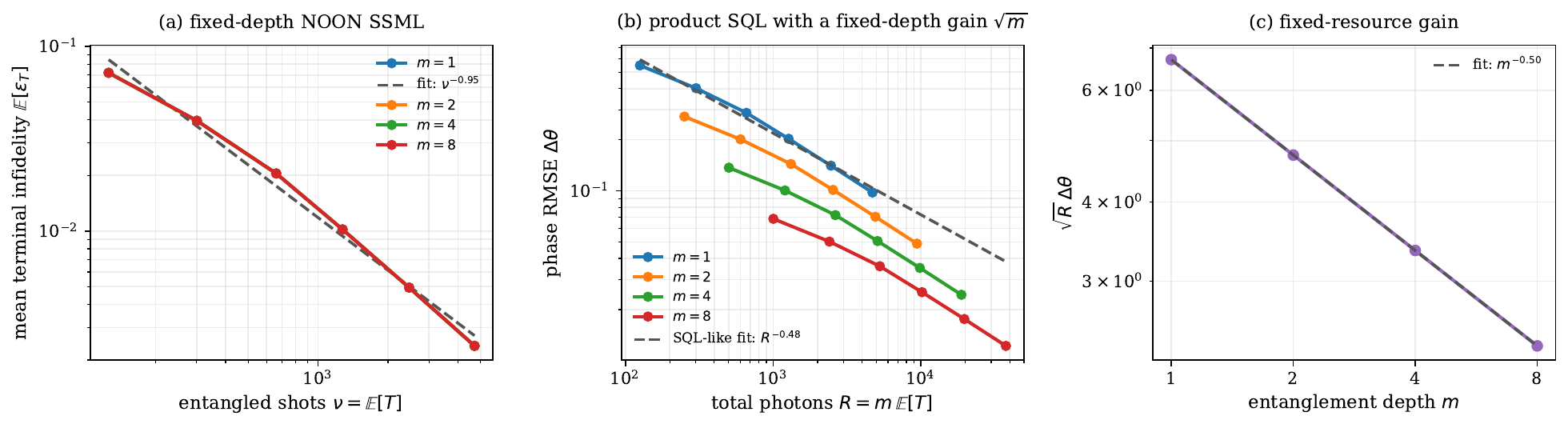}
    \caption{Numerical illustration for photonic NOON-state SSML. (a) Mean terminal infidelity versus entangled shots $\nu=\mathbb{E}[T]$ for $m=1,2,4,8$ in the local branch-resolved regime. The fit gives $\mathbb{E}[\epsilon_T]\propto \nu^{-0.95}$. (b) Phase RMSE versus total photon number $R=m\,\mathbb{E}[T]$. The slope is SQL-like, $\Delta\theta\propto R^{-0.48}$, while the entangled curves are shifted downward. (c) At fixed total resource, the prefactor obeys $\sqrt{R}\,\Delta\theta\propto m^{-1/2}$, confirming the preserved $\sqrt{m}$ advantage.}
    \label{fig:noon-main}
\end{figure}

\subsubsection*{Single-scale aliasing, multiscale hand-off, and why global resolution matters} %------------------------------------------------------------------------------------------------

The local fixed-depth simulation above deliberately starts inside the correct branch. That is the right setting for isolating the estimator's local Fisher-preserving behavior, but it suppresses the global ambiguity of a NOON fringe pattern. To expose that ambiguity, we repeat the simulation with a global prior at fixed $m=8$. The result is shown in Fig.~\ref{fig:aliasing}. The mean terminal infidelity keeps decreasing with total resource because the estimator finds a locally good fringe, but the global phase RMSE saturates because the estimator is often trapped on the wrong fringe. This makes explicit why low infidelity and low global parameter error are not equivalent for a single entangled scale.

\begin{figure}[t]
    \centering
    \includegraphics[width=0.90\textwidth]{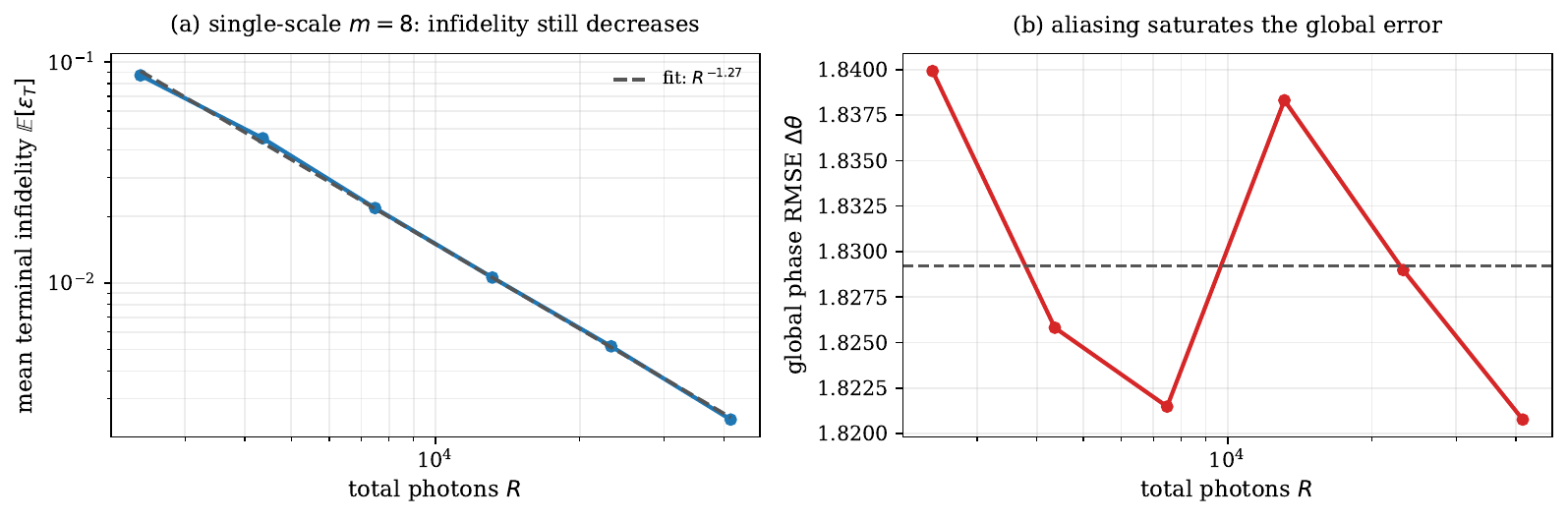}
    \caption{Single-scale aliasing under a global prior. (a) For a global prior and fixed $m=8$, the terminal infidelity continues to decrease with the total photon number. (b) The global phase RMSE, however, saturates because the protocol can lock onto an incorrect fringe. This is the metrological reason that the local fixed-depth numerical illustration of Fig.~\ref{fig:noon-main} is deliberately restricted to a branch-resolved regime.}
    \label{fig:aliasing}
\end{figure}

\begin{figure}[t]
    \centering
    \includegraphics[width=0.50\textwidth]{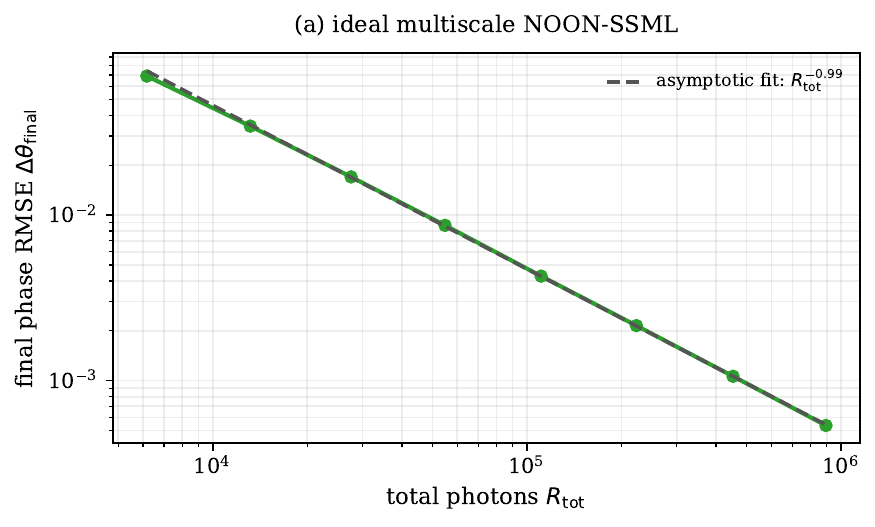}
    \caption{Ideal multiscale NOON-SSML. A coarse-to-fine architecture with stage depths $m_j=2^j$ and ideal branch-resolved hand-off exhibits the Heisenberg-compatible trend $\Delta\theta_{\rm final}\propto R_{\rm tot}^{-0.99}$ in the noiseless regime. The figure should be read as a resource-counting illustration rather than as a complete branch-selection protocol.}
    \label{fig:multiscale}
\end{figure}

Fig.~\ref{fig:multiscale} shows the complementary idealized multiscale simulation. Here, the stage $j$ uses depth $m_j=2^j$, and the stage output is handed to the next stage after ideal branch resolution. The resulting RMSE follows the asymptotic fit $\Delta\theta_{\rm final} \propto R_{\rm tot}^{-0.99}$, i.e., very close to the Heisenberg-compatible resource law predicted by Proposition~\ref{prop:multiscale-hl}. This result tells us that once global ambiguity is removed by a coarse-to-fine hand-off, nothing in the SSML estimator prevents the multiscale sensor from realizing the expected Heisenberg trend in the noiseless local regime.

\subsubsection*{Summary of the numerical analysis} %------------------------------------------------------------------------------------------------

Taken together, the three numerical datasets support the theoretical narrative of our work. The local fixed-depth dataset of Fig.~\ref{fig:noon-main} shows that the estimator preserves the expected $\nu^{-1}$ infidelity law and the $m^{-1/2}$ fixed-resource gain once branch ambiguity has been removed. Fig.~\ref{fig:aliasing} shows why one should not expect globally meaningful phase estimation from a single high-$m$ fringe pattern alone. Fig.~\ref{fig:multiscale} then shows how the same estimator layer becomes compatible with Heisenberg-like resource scaling when branch resolution is supplied stage by stage.

%What these simulations do not claim is equally important. They do not constitute a complete noise analysis, a universal optimal-control proof, or a general multiscale design for all sensing platforms. Their role is more focused: to demonstrate numerically that the conceptual distinctions developed in the analytical sections---certificate versus parameter error, Fisher preservation versus Fisher generation, fixed-depth gain versus multiscale Heisenberg compatibility---are visible in an explicit photonic NOON model.

%---------------------------------------------------------------------------------------------------------------------------------
\section*{Discussion and outlook}
%---------------------------------------------------------------------------------------------------------------------------------

The main message of this work can be stated succinctly. SSML can be cast as an adaptive estimator that makes an existing metrological gain operationally usable. In the sensing scenario, the terminal run $M_H$ is not merely a heuristic sign that the controller has settled, but is also an intrinsic sequential certificate whose scale is set by the length of the final success string. Near the optimum, this certificate becomes a QFI-calibrated parameter error bar. At the same time, the one-bit record that drives the protocol is locally Fisher-preserving: it carries the same local Fisher information as the underlying probe family.

This interpretation also sharpens the role of SSML inside quantum-enhanced sensing. The protocol is especially well matched to settings in which the experimental bottleneck lies not in heavy classical post-processing, but in online stabilization, repeated realignment, or rapid validation of a lock. In such settings a self-terminating estimator with a physically transparent halting meaning is a valuable architectural component. The present photonic NOON example illustrates this point concretely.

The picture of our work is conceptually clean. SSML can preserve QFI determined by a probe family while adding three things that are metrologically useful: an adaptive search policy, an interpretable monitored quantity, and a self-contained stopping rule. In this sense, SSML is not merely a learning algorithm transplanted into sensing. It is a minimal-feedback estimator layer that turns fidelity, Fisher information, and stopping-time evidence into a single operational object.

The setting in this work, i.e., pure-state probes, noiseless readout, and a clear distinction between branch acquisition and branch certification, was chosen to isolate the structural role of SSML as cleanly as possible. Relaxing these assumptions---for example, by including finite visibility, detector imperfections, or explicit branch-selection logic---is a natural direction for future work and would connect SSML more tightly to realistic quantum-sensor design.

%---------------------------------------------------------------------------------------------------------------------------------------------------------------------------------------------------------
\section*{Acknowledgement}
%---------------------------------------------------------------------------------------------------------------------------------------------------------------------------------------------------------

This work was supported by the Ministry of Science, ICT and Future Planning (MSIP) through the National Research Foundation of Korea (RS-2024-00432214, RS-2025-03532992, and RS-2025-18362970) and the Institute of Information and Communications Technology Planning and Evaluation grant funded by the Korean government (RS-2019-II190003, ``Research and Development of Core Technologies for Programming, Running, Implementing and Validating of Fault-Tolerant Quantum Computing System''), the Korean ARPA-H Project through the Korea Health Industry Development Institute (KHIDI), funded by the Ministry of Health \& Welfare, Republic of Korea (RS-2025-25456722), and the Ministry of Trade, Industry and Resources (MOTIR), Korea, under the project ``Industrial Technology Infrastructure Program'' (RS-2024-00466693). We acknowledge the Yonsei University Quantum Computing Project Group for providing support and access to the Quantum System One (Eagle Processor), which is operated at Yonsei University.

%---------------------------------------------------------------------------------------------------------------------------------
\section*{Methods}
%---------------------------------------------------------------------------------------------------------------------------------

\paragraph*{\bf Local fixed-depth dataset.}
For the main-text phase-sensing figure, we assume that a coarse acquisition stage has already identified the correct interference branch, so the initial metric coordinate $x_0=m(\theta-\tilde{\theta}_0)$ is sampled uniformly from $[-\pi/2,\pi/2]$. On a failure, the physical phase update is
\begin{eqnarray}
\tilde{\theta} \mapsto \tilde{\theta} + \frac{a}{m}(M_S+1)^{-b}r,
\quad
r \sim \mathrm{Uniform}[-1,1],
\label{eq:method-local-update}
\end{eqnarray}
which is equivalent to an $m$-independent metric update $x\mapsto x+a(M_S+1)^{-b}r$. This metric-coordinate formulation is useful because it isolates the estimator's role: the local landscape in $x$ is universal, while the conversion back to the physical phase keeps track of the entanglement depth through $\theta-\tilde{\theta}=x/m$. We set $a=0.3$ and $b=0.5$, use $10^4$ trials for each $(m,M_H)$, record the entangled-shot cost $\nu=\mathbb{E}[T]$, convert it to total photons $R=m\,\mathbb{E}[T]$, and estimate the log-log slopes by ordinary least squares.

\paragraph*{\bf Global single-scale aliasing dataset.}
To reveal the fringe ambiguity of a fixed-depth entangled sensor, we simulate a global prior $\theta-\tilde{\theta}_0\sim\mathrm{Uniform}[-\pi,\pi)$ at fixed $m=8$. The update rule is the same as in Eq.~\eqref{eq:method-local-update}, but the phase mismatch is wrapped modulo $2\pi$. This stylized model is sufficient to show the central point: terminal infidelity can continue to decrease even when the global phase RMSE no longer improves because the protocol may have locked onto an incorrect fringe.

\paragraph*{\bf Ideal multiscale dataset.}
For the multiscale main-text figure, stage $j$ uses entanglement depth $m_j=2^j$ with constant $M_H=320$. The hand-off to stage $j+1$ is implemented as $x^{\rm in}_{j+1}=2x^{\rm out}_j$ clipped to $[-\pi/2,\pi/2]$, which represents an ideal branch-resolved correction. This is intentionally a stylized resource-counting model rather than a complete branch-selection protocol. Its purpose is to isolate the asymptotic scaling predicted by Proposition~\ref{prop:multiscale-hl}, not to optimize a practical multiscale controller.

\appendix
\setcounter{figure}{0}
\renewcommand{\thefigure}{A\arabic{figure}}
\renewcommand{\theHfigure}{A\arabic{figure}}

%---------------------------------------------------------------------------------------------------------------------------------
\section{Run-length certification on a sensing manifold}
%---------------------------------------------------------------------------------------------------------------------------------

This appendix spells out in a more transparent way why the SSML halting rule is already a statistical statement. The idea is simple, but it is central enough to deserve a fuller discussion than the compressed derivation in the main text (for more details, see also Ref.~\cite{Bang2026run}).

%---------------------------------------------------------------------------------------------------------------------------------
\subsection{Why the terminal run is a valid object of inference}

In an adaptive protocol, one must normally be careful when turning outcomes into confidence statements, because the measurement settings may change from shot to shot. SSML has a special feature that resolves this concern at the moment of halting. During the final run of consecutive successes, no update is applied. The control is therefore fixed throughout the last $M_H$ shots. Whatever complicated adaptive history led to this point, the terminal run itself is generated by repeated Bernoulli sampling at one and the same success probability.

Let the final success probability be $p_s=1-\epsilon$. Then, the probability of observing a specific run of $M_H$ successes is simply
\begin{eqnarray}
\Prob(\underbrace{s\,s\,\cdots\,s}_{M_H\text{ times}}\mid\epsilon)=(1-\epsilon)^{M_H}.
\label{eq:app-run-prob}
\end{eqnarray}
Suppose one wishes to certify that the terminal infidelity is below some target value $\epsilon_0$. The null hypothesis is therefore $H_0:\epsilon\ge\epsilon_0$. Since the right-hand side of Eq.~(\ref{eq:app-run-prob}) decreases monotonically with $\epsilon$, the largest probability compatible with the null occurs at $\epsilon=\epsilon_0$. Hence,
\begin{eqnarray}
\sup_{\epsilon\ge\epsilon_0} \Prob(\underbrace{s\,s\,\cdots\,s}_{M_H}\mid\epsilon) = (1-\epsilon_0)^{M_H}.
\label{eq:app-np}
\end{eqnarray}
Rejecting the null at significance level $\eta$ means demanding that the right-hand side be at most $\eta$, which gives Eq.~(\ref{eq:eps-cert-exact}) of the main text.

This logic is exactly that of a one-sided sequential certificate. The point is not that SSML performs an externally designed hypothesis test after the fact. Rather, the protocol itself stores the relevant evidence in the variable it already tracks, namely the run length $M_S$. In that sense the certificate is intrinsic to the learning dynamics.

%---------------------------------------------------------------------------------------------------------------------------------
\subsection{Small-error expansion and the origin of the \texorpdfstring{$1/N$}{1/N} law}

The asymptotic form of the certificate follows from a one-line expansion. Since,
\begin{eqnarray}
\eta^{1/M_H}=e^{-\ln(1/\eta)/M_H},
\end{eqnarray}
we have
\begin{eqnarray}
1-\eta^{1/M_H} = \frac{\ln(1/\eta)}{M_H}-\frac{\ln^2(1/\eta)}{2M_H^2}+O(M_H^{-3}).
\label{eq:app-expansion}
\end{eqnarray}
Thus, the certificate scale is controlled by $1/M_H$ to leading order.

This observation also clarifies why the exponent $1$ repeatedly appears in SSML. Any protocol that halts on the basis of a run of $M_H$ identical outcomes can only certify a scale of order $1/M_H$, because that is the rate at which the null probability decays in the small-error regime. Since a halted trajectory necessarily uses at least $M_H$ copies, the terminal-run semantics itself excludes any parametrically faster certificate law than $O(1/N)$ in total consumed copies. Put differently, the $1/N$ law is already visible at the level of sequential evidence before one invokes Fisher information.

%---------------------------------------------------------------------------------------------------------------------------------
\subsection{From infidelity certificates to parameter certificates}

The remaining step is geometric rather than probabilistic. For a smooth pure-state family, the local mismatch between nearby states is measured by the QFI/Bures metric. Eq.~(\ref{eq:qfi-local-geometry}) says that, up to quartic corrections,
\begin{eqnarray}
\epsilon \simeq \frac{F_Q}{4}(\delta\lambda)^2.
\end{eqnarray}
Therefore, a certificate on $\epsilon$ immediately yields a certificate on $\delta\lambda$ after taking a square root. The parameter certificate of Proposition~\ref{prop:param-certificate} is precisely this conversion. The important conceptual point is that the halting rule and the local metric are complementary pieces of the same story: the run length determines how small the infidelity is certified to be, and the QFI tells us how that infidelity translates into a parameter error bar.

%---------------------------------------------------------------------------------------------------------------------------------
\section{Local QFI/Bures geometry and Fisher matching}
%---------------------------------------------------------------------------------------------------------------------------------

This appendix expands the information-geometric part of the argument and also explains why the monitored counter $M_S$ has direct metrological meaning.

%---------------------------------------------------------------------------------------------------------------------------------
\subsection{Local overlap and the QFI metric}

For a differentiable pure-state family $\ket{\psi_\lambda}$, the fidelity between neighboring states has the standard expansion~\cite{Braunstein1994,Paris2009}
\begin{eqnarray}
\abs{\braket{\psi_\lambda}{\psi_{\lambda+\delta\lambda}}}^2 = 1-\frac{F_Q(\lambda)}{4}(\delta\lambda)^2+O((\delta\lambda)^4),
\label{eq:app-qfi-expansion}
\end{eqnarray}
where
\begin{eqnarray}
F_Q(\lambda) =4\left(\braket{\partial_\lambda\psi_\lambda}{\partial_\lambda\psi_\lambda} - \abs{\braket{\psi_\lambda}{\partial_\lambda\psi_\lambda}}^2\right)
\end{eqnarray}
is the QFI. In a compensation family, the overlap entering the success probability is the overlap between the compensated state and the fiducial reference. Thus, the same local expansion applies to the return probability itself, leading to Eq.~(\ref{eq:qfi-local-geometry}) in the main text.

It is useful to define the dimensionless local coordinate
\begin{eqnarray}
y := \sqrt{F_Q}(\tilde{\lambda}-\lambda).
\end{eqnarray}
Then, the local infidelity is simply $1-p_s\simeq y^2/4$. This is the sense in which the return-probability landscape near the optimum is universal once distances are measured in the QFI metric.

%---------------------------------------------------------------------------------------------------------------------------------
\subsection{Detailed calculation of Fisher matching}

Starting from the local expansion
\begin{eqnarray}
p_s(x)=1-\frac{F_Q}{4}x^2+O(x^4),
\end{eqnarray}
we differentiate to obtain
\begin{eqnarray}
\partial_x p_s(x)=-\frac{F_Q}{2}x+O(x^3).
\end{eqnarray}
The numerator of the Bernoulli Fisher information is therefore
\begin{eqnarray}
[\partial_x p_s(x)]^2=\frac{F_Q^2}{4}x^2+O(x^4).
\end{eqnarray}
The denominator is
\begin{eqnarray}
p_s(x)[1-p_s(x)]
=\left(1-\frac{F_Q}{4}x^2+O(x^4)\right)
\left(\frac{F_Q}{4}x^2+O(x^4)\right)
=\frac{F_Q}{4}x^2+O(x^4).
\end{eqnarray}
Dividing the two gives
\begin{eqnarray}
I_{\rm cl}(x)=F_Q+O(x^2).
\end{eqnarray}
The local limit is therefore $I_{\rm cl}(0)=F_Q$.

This calculation shows that the apparent crudeness of one-bit readout is not the relevant quantity in the local regime. The success probability varies quadratically with the mismatch, and the derivative of that probability is exactly large enough to compensate for the Bernoulli variance. The resulting Fisher information is therefore finite and equal to the probe QFI to leading order.

%---------------------------------------------------------------------------------------------------------------------------------
\subsection{Why the monitored counter works}

The monitored proxy $\epsilon\approx (1+M_S)^{-1}$ can also be understood more transparently in the sensing language. Assume that over a short time window the control is nearly stationary so that the success probability is approximately fixed. Then, the number $L$ of consecutive successes before the next failure is a geometric random variable,
\begin{eqnarray}
\Prob(L=\ell\mid \epsilon)=(1-\epsilon)^\ell\epsilon \quad (\ell=0,1,2,\dots).
\end{eqnarray}
Its mean is
\begin{eqnarray}
\mathbb{E}[L\mid\epsilon]=\sum_{\ell\ge 0}\ell(1-\epsilon)^\ell\epsilon=\frac{1-\epsilon}{\epsilon}.
\end{eqnarray}
Thus, replacing the random run length $L$ by the observed counter $M_S$, we obtain the natural local estimate
\begin{eqnarray}
\epsilon\approx(1+M_S)^{-1}.
\end{eqnarray}
Combined with the local QFI/Bures relation, this yields Eq.~(\ref{eq:monitoring-proxy}) of the main text. In practice, the proxy is not exact because the control is still slowly adapting, but this derivation explains why it is so effective in both the numerical and the experimental SSML literature.

%---------------------------------------------------------------------------------------------------------------------------------
\section{Entangled probes and the meaning of quantum gain}
%---------------------------------------------------------------------------------------------------------------------------------

This appendix expands the scaling argument behind Corollary~\ref{cor:fixed-depth}. The main conceptual point is that ``quantum gain'' can mean different things depending on which resource is held fixed.

For the GHZ/NOON family of Eq.~(\ref{eq:noon-state}), one shot contains $m$ particles and has QFI $F_Q=m^2$. After $\nu$ independent entangled shots, the total Fisher information is $\nu m^2$. An efficient local estimator therefore has variance scaling as
\begin{eqnarray}
\Var(\tilde{\theta})=O\!\left(\frac{1}{\nu m^2}\right).
\end{eqnarray}
If one keeps $m$ fixed and increases the number of shots, the total particle budget is $R=m\nu$ and the RMSE behaves as
\begin{eqnarray}
\Delta\theta=O\!\left(\frac{1}{\sqrt{Rm}}\right).
\end{eqnarray}
Thus, the fixed-depth entanglement changes the prefactor at fixed total resource, but not the $R^{-1/2}$ slope. This is the reason why the main-text numerical result is a preserved quantum gain rather than a Heisenberg slope in $R$.

The distinction matters because it tells us what to expect from Fig.~\ref{fig:noon-main}(b). The SQL-like slope is not a failure of SSML. It is exactly the correct scaling law when the entanglement depth is kept fixed while the total particle number is varied. What the estimator must preserve in that setting is the downward shift of the RMSE curves with increasing $m$, and that is what the simulations show.

To obtain Heisenberg scaling in the total resource, the entanglement depth itself must grow with the stage index or with the overall resource. That is the role of the multiscale architecture, discussed in the next appendix section.

%---------------------------------------------------------------------------------------------------------------------------------
\section{Ideal multiscale architecture and resource counting}
%---------------------------------------------------------------------------------------------------------------------------------

The multiscale idea can be understood as a deterministic version of the acquisition--lock picture. A small-$m$ stage has coarse fringes and therefore a large capture range. A large-$m$ stage has narrow fringes and therefore high local sensitivity. Using both in sequence, we can resolve the tension between global identifiability and local precision.

Suppose that the stage $j$ uses entanglement depth $m_j=2^j$ and is run until the residual phase is reduced to order $c/m_j$ for some constant $c<\pi/2$. Because the stage then operates in a local basin, the run-length threshold required to certify this residual accuracy is only $O(1)$ in the idealized regime. Consequently, the number of entangled shots consumed at stage $j$ is $O(1)$, while the particle cost per shot is $m_j$. The stage cost is therefore
\begin{eqnarray}
R_j=O(m_j).
\end{eqnarray}
Summing over stages up to $J$, we attain the geometric total
\begin{eqnarray}
R_{\rm tot}=\sum_{j=0}^J O(m_j)=O(m_J).
\end{eqnarray}
The final residual after stage $J$ is $O(1/m_J)$, so the final RMSE is
\begin{eqnarray}
\Delta\theta_J=O(1/m_J)=O(R_{\rm tot}^{-1}).
\end{eqnarray}
This is the precise content of Proposition~\ref{prop:multiscale-hl}.

Of course, a practical implementation would still need an explicit and robust branch-selection rule for the hand-off between scales. The point of the proposition is not to hide that difficulty, but to separate it from the role of SSML itself. Once branch resolution is supplied---by a coarse stage, by prior information, or by a more elaborate acquisition scheme---the SSML estimator layer remains fully compatible with the asymptotic Heisenberg resource law.

\bibliography{ref_SSML-sensing}

\end{document}